\documentclass[twocolumn]{IEEEtran}    

\usepackage{graphicx}          

\usepackage{amsmath}
\usepackage{amsthm}
\usepackage{amssymb}
\usepackage{todonotes}
\usepackage{array}
\usepackage{cite}
\usepackage[utf8]{inputenc}
\usepackage[english]{babel}
\usepackage{float}
\usepackage{comment} 
\usepackage{subfig}
\usepackage{graphicx}
\usepackage[labelfont=bf]{caption}

\floatstyle{ruled}
\newfloat{algorithm}{tbp}{loa}
\providecommand{\algorithmname}{Algorithm}
\floatname{algorithm}{\protect\algorithmname}

\newtheorem{thm}{Theorem}
\newtheorem{example}{Example}
\newtheorem{cor}{Corollary}
\newtheorem{lem}{Lemma}

\newtheorem{prop}{Proposition}

\newtheorem{defn}{Definition}

\newtheorem{prob}{Problem}

\begin{document}
\global\long\def\ones{\mathds{1}}%
\global\long\def\G{\mathcal{G}}%
\global\long\def\E{\mathcal{E}}%
\global\long\def\V{\mathcal{V}}%
\global\long\def\M{\mathcal{M}}%
\global\long\def\Y{\mathcal{Y}}%
\global\long\def\U{\mathcal{U}}%
\global\long\def\C{\mathcal{C}}%
\global\long\def\T{\mathcal{T}}%
\global\long\def\graph{\mathcal{G}=\left(\mathcal{V},\mathcal{E}\right)}%
\global\long\def\TM{T_{\left(\mathcal{T},\mathcal{C}\right)}}%
\global\long\def\HT{\mathcal{H}_{2}}%
\global\long\def\R{\mathbb{R}}%
\global\long\def\diag{\mathrm{diag}}%
\global\long\def\dim{\mathrm{dim}}%
\global\long\def\trace{\text{Tr}}%
\global\long\def\Hinf{\mathcal{H}_{\infty}}%
\global\long\def\Tree{\mathbb{T}\left(\G\right)}%
\global\long\def\gcont{\G\sslash\pi}%
\global\long\def\N{\mathcal{N}}%


\title{Product Form of Projection-Based Model Reduction\\ and its Application to Multi-Agent Systems} 

\author{Noam Leiter and Daniel Zelazo,
\thanks{N. Leiter and D. Zelazo are with the Faculty of Aerospace Engineering, Israel Institute of Technology, Haifa, Israel.
    {\tt\small noaml@campus.technion.ac.il, dzelazo@technion.ac.il}.}
}
\maketitle

\begin{abstract}                          
Orthogonal projection-based reduced order models (PROM) are the output of widely-used
model reduction methods. In this work, a novel product form is derived
for the reduction error system of these reduced models, and it is
shown that any such PROM can be obtained from a sequence of 1-dimensional projection reductions. Investigating the error system product form, we then
define interface-invariant PROMs, model order reductions with projection-invariant input and output matrices, and it is shown that for such PROMs the error product systems are strictly proper. Furthermore, exploiting this structure, an analytic $\mathcal{H}_{\infty}$ reduction error bound is obtained and an $\mathcal{H}_{\infty}$ bound optimization
problem is defined.
Interface-invariant reduced models are natural to graph-based model
reduction of multi-agent systems where subsets of agents function as
the input and output of the system. In the second part of this study,
graph contractions are used as a constructive solution approach to
the $\mathcal{H}_{\infty}$ bound optimization problem for multi-agent
systems. Edge-based contractions are then utilized in a greedy-edge
reduction algorithm and are demonstrated for the model reduction of a
first-order Laplacian controlled consensus protocol. 

\end{abstract}

\section{Introduction}\label{sec.Intro}

Model-order reduction is an essential tool for the design and study
of large-scale systems introduced by modern technologies. Of particular interest is the study of model reduction for the design, simulation, and implementation  of controllers for large-scale systems. For example, optimal controllers for  linear systems are often at least the order of the physical system model
\cite{dullerud2013course}. In order to implement low-order controllers for large scale systems,
model reduction of the design model or full-order controller is commonly performed \cite{hyland1984optimal,leiter2019optimal}.

A widely-used family of reduced-order models are the \textit{projection-based
reduced order models} (PROMs). Well established PROM producing methods,
such as truncated balanced realizations, preserve stability, guarantee
minimality and provide \textit{a priori} reduction error bounds \cite{moore1981principal}.
These methods, however, may be unfeasible for very large-scale systems
due to their computational complexity \cite{benner2005model}. As a result, many works aimed at finding sub-optimal efficient solutions,
e.g. by alternating projection methods \cite{grigoriadis1996low}, or Krylov-subspace techniques which are computationally
efficient and suitable for extremely large-scale systems \cite{antoulas2000survey}. Such methods,
however, may fail to provide stable and minimal reduced order
realizations \cite{jaimoukha1997implicitly}.

An important instance of these large-scale systems are multi-agent systems.  
A great challenge in the study of multi-agent systems
is to find efficient reduction methods that guarantee stability and
assure minimality of their reduced order realizations, preferably,
with optimal or suboptimal reduction errors. For networked multi-agent
systems, the interconnection properties of the underlying graph can
be related in some cases to the stability, controllability and observability
of the system \cite{aguilar2014graph}. Therefore, it is desirable
that reduced-order models are found that preserve, in some sense, the network structure
of the full-order system, which cannot be performed with standard
unstructured model-reduction methods. In that direction, several recent
studies were performed. In \cite{monshizadeh2014projection,ishizaki2014model},
PROMs of consensus-type multi-agent models were considered based on
graph-contractions over vertex partitions. In \cite{jongsma2015model},
removal of cycle completing edges was suggested for model simplification of the consensus protocol. The reduction of second-order network systems with structure preservation using hierarchical $\mathcal H_2$ clustering was demonstrated in \cite{cheng2017reduction}. While these methods where limited to first or second-order multi-agent models, in \cite{Leiter2017a}, a more general graph-based model reduction was presented that preserves the functional structure of the multi-agent system. 
A framework for optimal structured model-order reduction of multi-agent systems was recently presented in \cite{yu2021h_2}.  Here, a convex relaxation technique was derived for the $\mathcal H_2$ model reduction of diffusively coupled second-order networks.

In this work, we reexamine the well known orthogonal PROMs and their realizations. The error system of PROMs is naturally  described with an augmented system realization, allowing the PROM reduction error to be evaluated with standard system performance metrics, such as the $\mathcal{H}_{\infty}$ and $\mathcal{H}_{2}$ norms. However, this error system realization usually does not provide any analytic insight, and various transformation techniques were derived to bring it to more useful forms. These include an upper triangular block structure \cite{ishizaki2014model}, allowing derivation of analytical bounds.  In this study we show that any orthogonal PROM error system can be presented as a product of three LTI systems, capturing the reduction effect of the input to state, state to output, and the internal dynamical structure. Investigating this error system product
form, an analytic $\mathcal{H}_{\infty}$ reduction error bound is
obtained and an $\mathcal{H}_{\infty}$ bound optimization problem
is defined as a relaxation of the optimal PROM problem. Furthermore,
it is shown that any orthogonal PROM can be obtained from a sequence
of singleton projections, projections from dimension $n$ to $n-1$. 

Of particular interest in this work are
\emph{interface invariant }PROMs (IIPROMs). These are reduced models
that maintain the input-output structure of the full order model. It is shown that for IIPROMs the error product systems are strictly
proper.
%
IIPROMs are
natural to multi-agent systems where a subset of agents serve as input
and output ports of the network.  These I/O ports may be interconnected with an
external controller and it is required, therefore, that any reduced model preserves this interface structure.  For this purpose, we propose an edge-based graph
contraction method and utilize it in a tree-based greedy-edge heuristic to solve the
PROM $\Hinf$ bound optimization problem. We then apply
an this graph contraction algorithm to obtain suboptimal $\Hinf$ IIPROMs of Laplacian
consensus systems.

The remaining sections of this paper are as follows. In Section \ref{sec:Problem-Formulation},
we formulate the optimal orthogonal PROM and IIPROM problems. In Section
\ref{sec:The-Product-Form}, the product form of orthogonal PROMs is derived. In Section \ref{sec:The-Orthogonal-PROM hinf bound},
the $\Hinf$ error bound is derived for PROMs and PROM sequences.
Section \ref{sec:Model-Reduction-of} presents model reduction of
multi-agent systems by graph contractions and the greedy-edge optimization
method. In Section \ref{sec:Case-Studies}, these results are demonstrated
with some numerical examples of model reduction of a multi-agent system,
and Section \ref{sec:Conclusions} provides concluding remarks.

\paragraph*{Notations}

The spectrum of a real matrix $A\in\mathbb{R}^{n\times n}$ is the
set of eigenvalues $\lambda\left(A\right)=\left\{ \lambda_{k}\left(A\right)\right\} _{k=1}^{n}$
where $\lambda_{k}\left(A\right)\in\mathbb{C}$ is the $k$th eigenvalue
of $A$ in ascending order, $\left|\lambda_{1}\right|\leq\left|\lambda_{2}\right|\leq\ldots\leq\left|\lambda_{n}\right|$.
The corresponding eigenvectors are $\left\{ u_{k}\left(A\right)\right\} _{k=1}^{n}$.
For a symmetric matrix we have an eigenvalue decomposition $A=P\left(A\right)\Lambda\left(A\right)P^\top\left(A\right)$,
where $P\left(A\right)=\left[u_{1}\left(A\right),u_{2}\left(A\right),\ldots,u_{n}\left(A\right)\right]$
is an orthonormal matrix and $\Lambda\left(A\right)=\diag\left(\lambda\left(A\right)\right)\in\mathbb{R}^{n\times n}$.
A symmetric matrix is \textit{positive-definite} if $\lambda_{i}\left(A\right)>0$
for $i\in\left[1,n\right]$ and is denoted as $A\succ0$. The $2$-norm of a matrix $A$ is $\|A \|_2=\max_i(\sqrt{\lambda_i(A^TA)})$.
For two
matrices $A$ and $B$, $\diag(A,B)$ is a block diagonal matrix with
$A$,$B$ on the diagonal. The entries of a matrix $A$ are denoted
$\left[A\right]_{ij}$. The $i$th Euclidean basis column vector is
denoted as $\mathbf{e}_{i}$. The Kronecker product is denoted by $\otimes$.

A graph $\mathcal{G}=\left(\mathcal{V},\E,\mathcal{W}\right)$ consists
of a vertex set $\V\left(\G\right)$, an edge set $\E\left(\G\right)=\{\epsilon_{1},\ldots,\epsilon_{|\E|}\}$
with $\epsilon_{k}\in\V^{2}$, and a set of edge weights, $\mathcal{W}\left(\G\right)=\{w_{1},\ldots,w_{|\E|}\}$ with $w_i \in \mathbb{R}$.
The order of the graph is defined as the number of nodes, $|\V|$. Two nodes
$u,v\in\V\left(\G\right)$ are \textit{adjacent} if they are the endpoints
of an edge $\left\{ u,v\right\} $, and we denote this by $u\sim v$.
If $\G$ is an \textit{undirected graph} then the head and tail of
each edge are arbitrary. A \textit{self-loop} is an edge were the head and tail are the same node,
and \textit{duplicate edges} are any pair $\epsilon_{i},\epsilon_{j}\in\E$
with the same head and tail nodes.
A \textit{simple graph} does not include self-loops or duplicate edges. A \textit{multi-graph}
is a graph that includes duplicate edges.

\section{Problem Formulation \label{sec:Problem-Formulation}}

An LTI system $\Sigma$ with realization $\Sigma:=\left(A,B,C,D\right)$
is the dynamical system, 
\begin{equation}
\begin{cases}
\dot{x}(t) & =Ax(t)+Bu(t)\\
y(t) & =Cx(t)+Du(t)
\end{cases},\label{eq:LTI model}
\end{equation}
where $x(t)\in\mathbb{R}^{n_{x}}$ is the system state, $u(t) \in\mathbb{R}^{n_{u}}$
are the inputs and $y(t) \in\mathbb{R}^{n_{y}}$ are the
outputs. The matrices $A\in\R^{n_{x}\times n_{x}}$, $B\in\R^{n_{x}\times n_{u}}$,
$C\in\R^{n_{y}\times n_{x}}$ and $D\in\R^{n_y\times n_u}$ are the system
parameters. 
The corresponding
transfer function matrix (TFM) representation of $\Sigma$ is given as
\begin{equation}
\hat{\Sigma}\left(s\right)=D+C\left(sI_{n}-A\right)^{-1}B.\label{eq:TFM}
\end{equation}
Hereafter, the notation $\hat{\Sigma}$ will be used, without the explicit dependence on $s$, to denote the TFM of a system $\Sigma$.

A realization
$\Sigma:=\left(A,B,C,D\right)$ is \textit{minimal} if it is controllable
and observable. The order of a system is its \textit{McMillan degree},
denoted as $\deg\left(\Sigma\right)$, which is the order of any minimal
realization of $\Sigma$ \cite{de2000minimal}.

We consider also \emph{multi-agent systems} (MAS) as a set of $n$ agents  that interact with each other over a network described by a graph $\G$.  We assume further that in such a network, a subset of the agents may be subject to external control inputs, and a subset may be accessed for measurements.  Formally, we denote  the \emph{input nodes} by  the set $\U\subseteq\V\left(\G\right)$ with $|\U| =m$, and the set of \emph{output nodes} by the set $\Y\subseteq\V\left(\G\right)$ with $|\Y| = p$.  In the case of linear MAS, we then have the realization
\begin{equation}
\hspace{-.3cm}\Sigma\left(\G,\mathcal{U},\Y\right):=\left(A\left(\G\right),B\left(\G,\mathcal{U}\right),C\left(\G,\Y\right),D\left(\mathcal{U},\Y\right)\right),\label{eq:controlled multi-agent system}
\end{equation}
where the matrices $A\left(\G\right)\in\R^{n_{x}\times n_{x}}$, $B\left(\G,\mathcal{U}\right)\in\R^{n_{x}\times n_{u}}$,
$C\left(\G,\Y\right)\in\R^{n_{y}\times n_{x}}$ and $D\left(\U,\Y\right)\in\R^{n_{y}\times n_{u}}$
are the system matrices as a function of the underlying graph structure, with $n_{x}=d_{x}\times n$, $n_{u}=d_{u}\times m$ and $n_{y}=d_{y}\times p$. Here, $d_x,d_u,$ and  $d_y$ represent the dimension of the state, input, and output of each agent in the network.  This is system is visualized in Figure \ref{fig:An interface multi-agent system}.  Note that MAS models of this form include classical setups such as diffusively coupled networks \cite{burger2014duality} and Laplacian dynamics \cite{mesbahi2010graph}. 


\begin{center}
\begin{figure}
\centering{}\includegraphics[width=2.5in]
{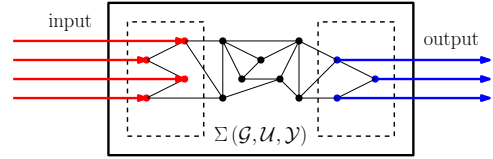}
\caption{A controlled  MAS with interface agents, red nodes identify the set $\mathcal{U}$, blue nodes the set $\mathcal{Y}$.\label{fig:An interface multi-agent system}}
\end{figure}
\end{center}

A \textit{reduced-order model} of $\Sigma$ is any system with realization $\Sigma_{r}:=\left(A_{r},B_{r},C_{r},D_{r}\right)$
mapping $u\left(t\right)\mapsto y_{r}\left(t\right)$, with $u\left(t\right)\in\mathbb{R}^{n_{u}}$
and $y_{r}\left(t\right)\in\mathbb{R}^{n_{y}}$, such that $\deg\left(\Sigma_{r}\right)<\deg\left(\Sigma\right)$.
Reduction error analysis can be performed by constructing an augmented
error system, 
\begin{equation}
\Sigma_{e}=\Sigma-\Sigma_{r},\label{eq: reduction error system}
\end{equation}
with realization $\Sigma_{e}:=\left(A_{e},B_{e},C_{e},D_{e}\right)$,
where $x_{e}(t)=\begin{bmatrix}x^\top(t)&x_{r}^\top(t)\end{bmatrix}^\top$, $y_{e}(t)=y(t)-y_{r}(t)$, $A_{e}=\diag(A,A_{r})$,
$B_{e}=\begin{bmatrix}B^\top& B_{r}^\top\end{bmatrix}^\top$, $C_{e}=\begin{bmatrix}C & -C_{r}\end{bmatrix}$
and $D_{e}=D-D_{r}$.
The reduction error can then  be quantified using any system norm $\| \Sigma_{e}\|$. The two
most studied model reduction system norms are the $\HT$-norm and
the $\Hinf$-norm \cite{doyle1989state}, of which the $\Hinf$-norm
is the focus of this work. The $\Hinf$-norm of a stable proper system $\Sigma$ is given as
\begin{equation}
\| \Sigma\| _{\Hinf}=\sup_{\omega\in\R}\bar{\sigma}\left(\hat\Sigma\left(j\omega\right)\right),
\end{equation}
where $\bar{\sigma}(M)$ is the largest singular value of the matrix $M$.

A widely used family of reduction methods are \textit{projection-based
reduction}s. Given a system with realization $\Sigma:=\left(A,B,C,D\right)$ of order
$n$, a \textit{projection-based reduced order model} (PROM) is a system $\Sigma_{r}:=\left(P^\top AV,P^\top B,CV,D\right)$,
for any two matrices $P,V\in\mathbb{R}^{n\times r}$ such that $P^\top V=I_{r}$
\cite{gallivan2004sylvester}. If in addition $P=V$, the PROM is termed \emph{orthogonal}, i.e.,
\begin{equation}
\Sigma_{r}\left(\Sigma,P\right):=\left(P^\top AP,P^\top B,CP,D\right).\label{eq:PROM}
\end{equation}
Hereafter, all PROMs referred to in this study are orthogonal.

In this work we will examine a special class of PROMs that we term \emph{interface-invariant PROMs }. These are reduced models that maintain the  input-output  structure  of  the  full  order  model under the projection operation. Such PROMs are required, for example, for the reduction of controlled MAS (\ref{eq:controlled multi-agent system}) where the interface agent structure is maintained in the reduced model. We will also show that IIPROMs arise naturally when examining the error system (\ref{eq: reduction error system}) of PROMs.  
We now define formally the notion of an interface-invariant PROM (IIPROM).
\begin{defn}[IIPROM]\label{IIPROM} Given a system with realization
$\Sigma:=\left(A,B,C,D\right)$, an IIPROM of $\Sigma$ is any PROM $\Sigma_{r}\left(\Sigma,P\right):=\left(P^\top AP,P^\top B,CP,D\right)$
such that $C=CPP^\top$ and $B=PP^\top B$. \end{defn}
Note that an IIPROM does not require $PP^\top=I$, e.g., for $C=B^\top=[1\,0\,0]$ we can choose 
$$P=\frac{1}{\sqrt{2}}\left[\begin{array}{ccc} 1 & 1\\ 1 &  -1 \\ 0 & 0
\end{array}\right].$$ 
With the above notions in place, we can now formally state the optimal IIPROM problem.

\begin{prob}[optimal IIPROM]\label{optimal IIPROM}
Consider a stable proper system of order $n$ with realization $\Sigma:=\left(A,B,C,D\right)$ .
 Find $P\in\mathbb{R}^{n\times r}$ with $P^\top P=I_{r}$ such that
the PROM $\Sigma_{r}\left(\Sigma,P\right)$ \eqref{eq:PROM}
minimizes the $\Hinf$-norm of the reduction error system \eqref{eq: reduction error system}
and is interface-invariant, i.e., 
\begin{align}
\min_{P\in\mathbb{R}^{n\times r}} & \| \Sigma_{e}\| _{\Hinf}\label{eq:minimize the reduction error}\\
s.t.\,\,\, 
 & P^\top P=I_{r}\nonumber \\
 & C=CPP^\top\nonumber \\
 & B=PP^\top B.\nonumber 
\end{align}
\end{prob}
The constraints $P^\top P=I_{r}$, $C=CPP^\top$ and $B=PP^\top B$ make Problem \ref{optimal IIPROM} non-convex, and  there is, in general,
no closed-form or computationally efficient solution. In the following section we investigate the error system structure of IIPROMs and derive an IIPROM $\Hinf$ error upper bound. This bound will then be utilised for obtaining suboptimal solutions of Problem \ref{optimal IIPROM}. 

\section{The PROM $\protect\Hinf$ Bound\label{sec:The-Orthogonal-PROM hinf bound}}

The error system of PROMs can be described  with the augmented system realization (\ref{eq: reduction error system}) of dimension $\dim(x)+\dim(x_r)$. In the following section we show that any orthogonal PROM error system can be presented as a product of three appropriately defined LTI systems, two of dimension $\dim(x_r)$ and the third of dimension $\dim(x)$. This product form is then applied for the derivation of a PROM $\protect\Hinf$ error bound.  

\subsection{The Product Form of Orthogonal PROMs\label{sec:The-Product-Form}}
The following theorem presents a PROM error system product form, and this new error system structure will allow us to derive an $\Hinf$ bound for the PROM error system.

\begin{thm}[PROM error system product form]\label{PROM error system} Let $\Sigma:=(A,B,C,D)$,
and consider a PROM $\Sigma_{r}(\Sigma,P)$.
Then the error reduction system $\Sigma_{e}(\Sigma,P)=\Sigma_{r}(\Sigma,P)-\Sigma$
has the TFM 
\begin{equation}
\hat{\Sigma}_{e}(\Sigma,P)=C\Phi^{-1}Q(Q^\top\Phi^{-1}Q)^{-1}Q^\top\Phi^{-1}B,\label{eq:PROM error TFM}
\end{equation}
where $\Phi\triangleq sI_{n}-A$, and $Q$ is any projection such that $Q^\top P=0$ and 
$Q^\top Q=I_{n-r}$. 
Furthermore, $\Sigma_{e}$
can be expressed as the product of three systems,
\begin{equation}
\Sigma_{e}(\Sigma,P)=\Theta(\Sigma,P)\Delta(\Sigma,P)\Gamma(\Sigma,P),\label{eq:PROM error realization}
\end{equation}
with realizations 
\begin{align}
\Theta(\Sigma,P)&:=(A_{PP},A_{PQ},CP,CQ)\label{lti.1},\\
\Gamma(\Sigma,P)&:=(A_{PP},P^\top B,A_{QP},Q^\top B)\label{lti.2},\\
\Delta(\Sigma,P)&:=(A,Q,Q^\top,0_{p\times m}),\label{lti.3}
\end{align}
where $A_{PP}\triangleq P^\top AP$,
$A_{PQ}\triangleq P^\top AQ$ and $A_{QP}\triangleq Q^\top AP$.
\end{thm}

The proof of Theorem \ref{PROM error system} is given in the Appendix.

Investigating the error systems $\Theta(\Sigma,P)$
and $\Gamma(\Sigma,P)$ we observe that if the PROM
is an IIPROM, then its realization is strictly-proper.

\begin{cor}[IIPROM error system]\label{Interface-invariant error system}
If $\Sigma_{r}(\Sigma,P):=\left(P^\top AP,P^\top B,CP,D\right)$
is an IIPROM, then $\Theta(\Sigma,P)$ in \eqref{lti.1}
and $\Gamma(\Sigma,P)$ in \eqref{lti.2}
are strictly-proper with realizations
\begin{equation}
\Theta(\Sigma,P):=\left(A_{PP},A_{PQ},CP,0_{p\times m}\right),
\end{equation}
and
\begin{equation}
\Gamma(\Sigma,P):=\left(A_{PP},P^\top B,A_{QP},0_{p\times m}\right).
\end{equation}
\end{cor}
\begin{proof}Since $\Sigma_{r}(\Sigma,P)$ is an
IIPROM, from Definition \ref{IIPROM} we have
$C=CPP^\top$ and $B=PP^\top B$. From $Q^\top P=0$
and $Q^\top Q=I_{n-r}$ we get $PP^\top=I_{n}-Q^\top Q$
such that $C=C\left(I_{n}-QQ^\top\right)$ and
$B=\left(I_{n}-QQ^\top\right)B$, therefore, $CQ=0$
 and $Q^\top B=0$ obtaining our desired result.  $\qed$ 
\end{proof}
\subsection{The PROM Error System Bound}

The PROM error system is the product of three systems \eqref{eq:PROM error realization}, and we will make use of this form to derive an $\Hinf$ reduction error upper bound. 

\begin{prop}[PROM error bound]\label{PROM error bound} Let $\Sigma:=\left(A,B,C,D\right)$ with $A$ Hurwitz, and consider a PROM $\Sigma_{r}(\Sigma,P):=\left(P^\top AP,P^\top B,CP,D\right)$.
Then the $\Hinf$ norm of the error reduction system \eqref{eq: reduction error system}
is bounded as 
$$\| \Sigma_{e}(\Sigma,P)\| _{\Hinf}\leq b(\Sigma,P),$$
where 
\begin{equation}
b(\Sigma,P)=\| \Theta(\Sigma,P)\| _{\Hinf}\|\Delta(\Sigma,P)\| _{\Hinf}\| \Gamma(\Sigma,P)\| _{\Hinf} .\label{eq:PROM error bound equation}
\end{equation}
\end{prop}
\begin{proof} 
%
The proof follows directly from the submultiplicative of the $\Hinf$-norm  applied to \eqref{eq:PROM error realization}. $\qed$
\end{proof}

 For linear systems \eqref{eq:LTI model} with a symmetric matrix $A$, such as the controlled-consensus multi-agent systems studied in Section \ref{sec:Model-Reduction-of}, we can simplify the calculation of the bound \eqref{eq:PROM error bound equation}. The following Lemma, proven in \cite[Appendix A]{liu1998model} will be used for the derivation of Corollary \ref{PROM symmetric error bound}, presenting this simplified bound. 

\begin{lem}[\cite{liu1998model}]\label{symmetric w0 Hinf}Let $A$ be symmetric and Hurwitz.
Then for any $B$ of appropriate dimension,
\begin{equation}
\| B^\top\left(sI-A\right)^{-1}B\| _{\Hinf}=\| B^\top  A^{-1}B\| _{2}.\label{eq:B'*Ainv*B}
\end{equation}
\end{lem}

\begin{figure*}
\centering
\subfloat[][]{\includegraphics[width=.75\columnwidth]{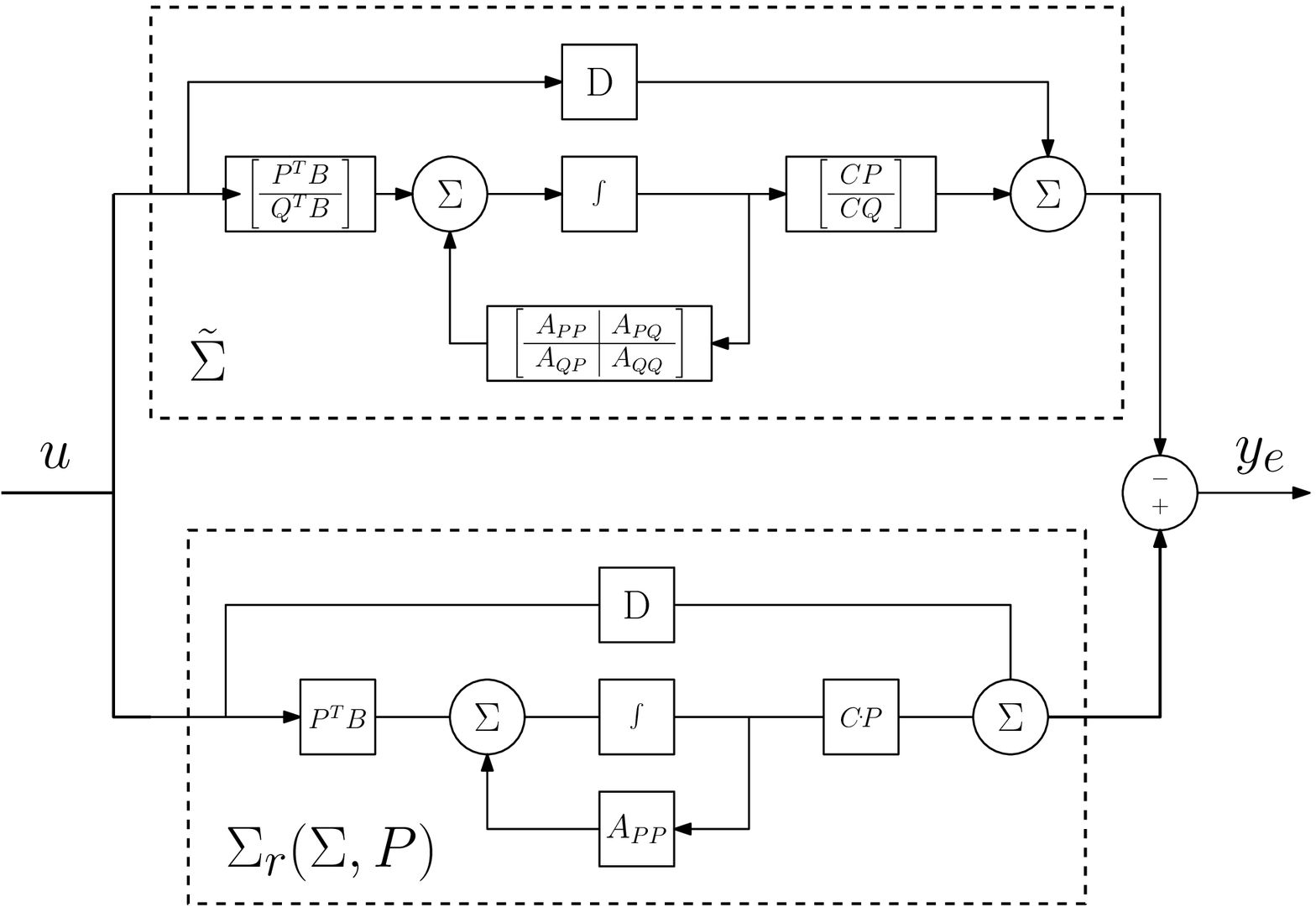}\label{fig: error system block diagram additive}}
\,\,\,\,\,
\subfloat[][]{\includegraphics[width=.75\columnwidth]{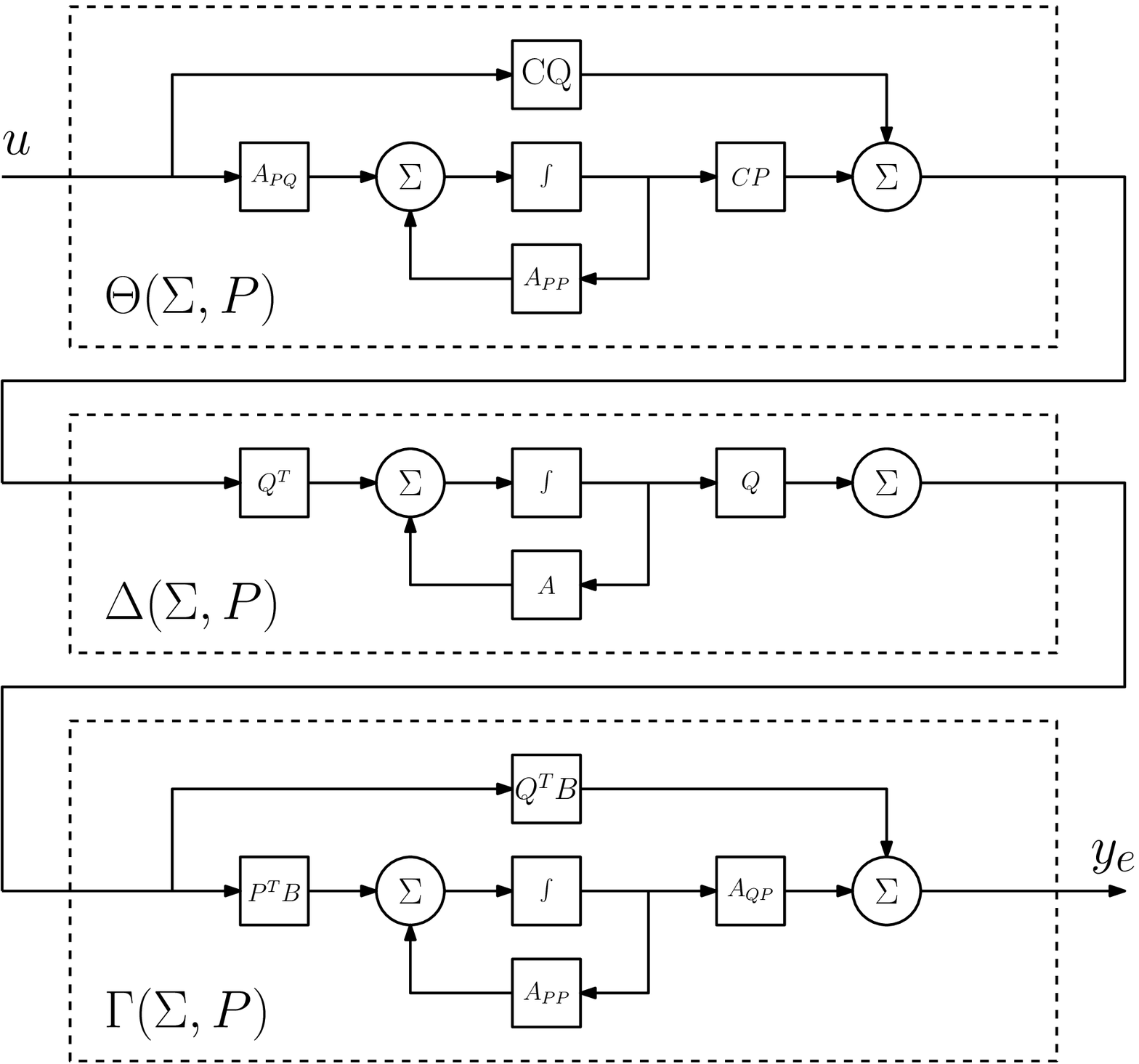}\label{fig: error system block product}}
\caption{The PROM error system block diagram in (a) additive form and (b) product form \eqref{PROM error system}.}\label{fig: error system block diagram}
\end{figure*}

\begin{cor}\label{PROM symmetric error bound} Let $\Sigma:=\left(A,B,C,D\right)$ with $A$
symmetric and Hurwitz, and consider a PROM $\Sigma_{r}(\Sigma,P):=\left(P^\top AP,P^\top B,CP,D\right)$.
Then the $\Hinf$ norm of the error reduction system \eqref{eq: reduction error system}
is bounded as 
$$\| \Sigma_{e}(\Sigma,P)\| _{\Hinf}\leq b(\Sigma,P),$$
where 
\begin{equation}
b(\Sigma,P)=\| \Theta(\Sigma,P)\| _{\Hinf}\| \Gamma(\Sigma,P)\| _{\Hinf}\| Q^\top A^{-1}Q\| _{2}.\label{eq:PROM symmetric error bound equation}
\end{equation}
\end{cor}

\begin{proof} Applying Lemma \ref{symmetric w0 Hinf}, we get $\|\Delta(\Sigma,P)\| _{\Hinf}=\| Q^\top A^{-1}Q\| _{2}$ and substituting it in \eqref{eq:PROM error bound equation} we obtain \eqref{eq:PROM symmetric error bound equation}.  $\qed$
\end{proof}

The PROM error bound is the product of the $\Hinf$-norms of the three LTI systems \eqref{lti.1}-\eqref{lti.3} constructing the error system. We observe that with the unitary transformation $\tilde\Sigma:=\left(U^\top AU,U^\top B,CU,D\right)$ with $U=\begin{bmatrix}P & Q \end{bmatrix}$, the full-order system $\tilde \Sigma$ has a realization
\begin{equation*}
    \tilde{\Sigma}:=\left(\left[\begin{array}{c|c}
A_{PP} & A_{PQ}\\
\hline A_{QP} & A_{QQ}
\end{array}\right],\left[\begin{array}{c}
P^{T}B\\
\hline Q^{T}B
\end{array}\right],\left[\begin{array}{c|c}
CP & CQ\end{array}\right],D\right),
\end{equation*}
where the PROM is 
\begin{equation*}
    \Sigma_{r}(\Sigma,P):=  \left(A_{PP},P^\top B,CP,D\right),
\end{equation*} and the block-diagram of $\Sigma_r-\tilde{\Sigma}$ in additive form is shown in Figure \ref{fig: error system block diagram additive}. 
If $A_{PQ}$, $A_{QP}$, $A_{QQ}$, $CQ$ and $Q^TB$ are all zeros, then $\Sigma_r-\Sigma=0$ (and in this case $\Sigma:=(A,B,C,D)$ is not a minimal realization). If the reduction error is not zero, and each of the three product systems implicitly captures the contribution to the reduction error of these parts of $\tilde{\Sigma}$  left out in $\Sigma_{r}$. The map $\Theta(\Sigma,P)$ captures  $A_{PQ}$ and $CQ$, $\Gamma(\Sigma,P)$ captures $A_{QP}$ and $Q^TB$ and $\Delta(\Sigma,P)$ captures $A_{QQ}$ (Figure \ref{fig: error system block product}).

Since there are no closed-form solutions to the optimal IIPROM Problem \ref{optimal IIPROM}, this structure suggests that minimizing the three reduction error contributions can provide good PROMs. As a first step in obtaining a suboptimal solution, we define the following suboptimal IIPROM problem that attempts to minimize the error reduction upper bound derived in Proposition \ref{PROM error bound}.

\begin{prob}[suboptimal IIPROM]\label{optimal IIPROM bound}
Consider a stable proper system of order $n$ with realization $\Sigma:=\left(A,B,C,D\right)$.
 Find $P\in\mathbb{R}^{n\times r}$ with $P^\top P=I_{r}$ such that
the PROM $\Sigma_{r}\left(\Sigma,P\right)$ \eqref{eq:PROM}
minimizes the reduction error bound \eqref{eq:PROM error bound equation}
and is interface-invariant, i.e., 
\begin{align}
\min_{P\in\mathbb{R}^{n\times r}} & b(\Sigma,P)\\
s.t.\,\,\, 
 & P^\top P=I_{r}\nonumber \\
 & C=CPP^\top\nonumber \\
 & B=PP^\top B.\nonumber 
\end{align}
\end{prob}

The following simple example provides a comparison between the solutions for the optimal IIPROM Problem \ref{optimal IIPROM} and the optimal IIPROM bound Problem \ref{optimal IIPROM bound}.

\begin{example}
Consider the SISO system $\Sigma:=\left(A,B,C,D\right)$
where 
\begin{align*}
    A = \begin{bmatrix}-2 & 1 & 0\\
1 & -2 & 1\\
0 & 1 & -1 \end{bmatrix},\, B=\begin{bmatrix}1 \\ 0 \\0 \end{bmatrix},\,C = \begin{bmatrix} 1 & 0 & 0\end{bmatrix},\, D=0.
\end{align*}
The corresponding TF is 
\begin{equation*}
    \hat{\Sigma}=\frac{s^{2}+3s+1}{s^{3}+5s^{2}+6s+1}
\end{equation*}
and $\|\Sigma\| _{\Hinf}=1$.
We observe that all matrices $P\in\mathbb{R}^{3\times2}$
complying with $P^\top P=I_2$, $CPP^\top=C$ and $PP^\top B=B$, can be parameterized by a scalar $\alpha\in\left[-1,1\right]$ in the following form, $$P\left(\alpha\right)=\begin{bmatrix}
1 & 0 \\ 0 & \alpha \\ 0 & \beta(\alpha)
\end{bmatrix},$$ 
where $\beta\left(\alpha\right)=\sqrt{1-\alpha^{2}}$.
All matrices $Q\in\mathbb{R}^{3\times1}$ such that $P^\top Q=0$ are
parametrized by $$Q\left(\alpha\right)=\begin{bmatrix}
0 \ -\beta\left(\alpha\right) & \alpha\end{bmatrix}^\top.$$ All IIPROMs $\Sigma_{r}\left(\Sigma,P\right)$ are then parameterized also by $\alpha$, such that the matrices of the product form system realizations of Theorem \ref{PROM error system} are
\begin{equation*}
    A_{PP}\left(\alpha\right) =\begin{bmatrix} -2 & \alpha \\ \alpha & -\alpha^2-(\beta(\alpha)-\alpha)^2 \end{bmatrix}
\end{equation*}
and
\begin{equation*}
    A_{PQ}\left(\alpha\right)=\begin{bmatrix} -\beta(\alpha) \\ 2\alpha^2+\alpha \beta(\alpha)-1 \end{bmatrix}=A_{QP}(\alpha)^\top.
\end{equation*}

The PROM TF is 
\begin{equation*}
    \hat{\Sigma}_{r}=\frac{s+\gamma^{2}}{s^{2}+\left(\gamma^{2}+2\right)s+2\gamma^{2}-\alpha^{2}},
\end{equation*}
and the product systems \eqref{lti.1}-\eqref{lti.3} TFs are
\begin{equation*}
    \hat{\Theta}=\hat{\Gamma}=\frac{-\beta s+\alpha^{3}-\alpha^{2}\beta+\alpha\beta^{2}-\beta^{3}}{s^{2}+\left(\alpha^{2}+\gamma^{2}+2\right)s+\alpha^{2}+2\gamma^{2}}
\end{equation*}
and
\begin{equation*}
    \hat{\Delta}=\frac{\left(\alpha^{2}+\beta^{2}\right)s^{2}+\left(2\gamma^{2}+2\alpha^{2}+\beta^{2}\right)s+2\gamma^{2}+\alpha^{2})}{s^{3}+5s^{2}+6s+1},
\end{equation*}
where $\gamma\triangleq\sqrt{\alpha^{2}+\left(\alpha-\beta\right)^{2}}$.

The PROM reduction error and reduction error bound are plotted in
Figure \ref{fig:error and bound of the second order system}. It is observed that the solution of the optimal IIPROM problem (Problem \ref{optimal IIPROM}) is  $\min_\alpha\|\Sigma-\Sigma_r\|_\infty=0.03$ obtained for $\alpha^{*}=0.76$. The solution of the optimal IIPROM bound problem (Problem \ref{optimal IIPROM bound}) is $\min_\alpha b(\Sigma,P)=0.039$ obtained for $\alpha^{**}=0.73$. We observe that the optimal IIPROM error is close to the sub-optimal bound.

\begin{figure}
\centering
\includegraphics[width=.65\columnwidth]{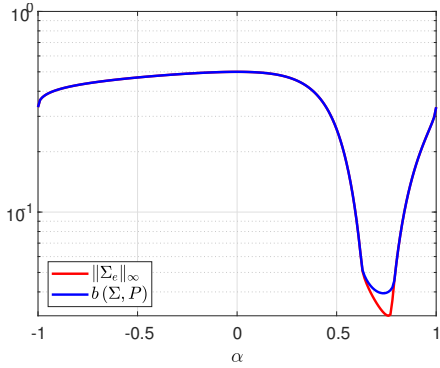}\caption{The IIPROM reduction error and bound of the second order system as a function of the projection parameter $\alpha$.}\label{fig:error and bound of the second order system}
\end{figure}

\end{example}

\subsection{The PROM Sequence Bound}

In the following subsection we present a Lemma showing that any projection from $\R^{n}$ to $\R^{r}$  can be obtained from a sequence of $n-r$ projections, each reducing the dimension by one.  We  denote such projections as \textit{singleton projections}. This sequential projection representation is then utilized to obtain a sequential PROM bound that is useful for obtaining sub-optimal solutions to Problem \ref{optimal IIPROM bound}.

\begin{lem}\label{projection product lemma}Let $P\in\R^{n\times r}$  with
$P^\top P=I_{r}$ be a projection for $r<n$. Then
there exists a sequence $\left\{ P_{\left(k\right)}\right\} _{k=1}^{n-r}$
with $P_{\left(k\right)}\in\R^{n-k+1\times n-k}$, $P_{\left(k\right)}^\top P_{\left(k\right)}=I_{n-k}$
such that $P=\Pi_{k=1}^{n-r}P_{\left(k\right)}$. \end{lem}
\begin{proof}Let $P\in\R^{n\times r}$, $P^\top P=I_{r}$,
then there exists $Q\in\R^{n\times n-r}$ with $Q^\top Q=I_{n-r}$ such
that $P^\top Q=0$. Construct $P_{\left(1\right)}=\left[P,q_{1},q_{2},\ldots,q_{n-r-1}\right]$, where $q_{i}$ is the $i$th column of $Q$,  and $P_{\left(k\right)}=\left[I_{n-k},0_{n-k\times1}\right]^\top$ for $k\in\left[2,n-r\right]$. Since $P^\top P=I_{r}$ and $P^\top Q=0$ we have   $P_{\left(1\right)}^\top P_{\left(1\right)}=I_{n-1}$, and for $k\in\left[2,n-r\right]$ it is trivial that  $P_{\left(k\right)}^\top P_{\left(k\right)}=I_{n-k}$. From the construction we get
\begin{align*}
    \left(\prod_{k=1}^{n-r}P_{\left(k\right)}\right)&=P_{\left(1\right)}\left(\prod_{i=1}^{n-r-1}P_{\left(1+i\right)}\right)\\&=P_{\left(1\right)}\left(\prod_{i=1}^{n-r-1}\left[I_{n-i},0_{n-1\times1}\right]^\top\right)\\&=P  .
\end{align*}
$\qed$
\end{proof}

Note that infinite other sequences $P=\Pi_{k=1}^{n-r}\tilde{P}_{\left(k\right)}$ can be produced by the transformations $\tilde{P}_{\left(1\right)}=P_{\left(1\right)}U_{\left(1\right)}$,
$\tilde{P}_{\left(k\right)}=U_{\left(k-1\right)}^\top P_{\left(k\right)}U_{\left(k\right)}$ for $k\in\left[2,n-r-1\right]$, and $\tilde{P}_{\left(n-r\right)}=U_{\left(n-r-1\right)}^\top P_{\left(n-r\right)}$, where $U_{\left(k\right)}\in \R^{n-k\times n-k}$ is an orthogonal matrix.

By expressing a PROM as a sequence of singleton projections, we obtain
the following PROM sequence bound.

\begin{prop}[Singleton PROM sequence bound]Let $\Sigma$ be an LTI
system with realization $\left(A,B,C,D\right)$ with $A$ Hurwitz,
and consider the PROM, $\Sigma_{r}:=(P^\top A P,P^\top B,CP,D)$,
and let $\left\{ P_{\left(k\right)}\right\} _{k=1}^{n-r}$ be a sequence with
$P_{\left(k\right)}\in\R^{n-k+1\times n-k}$, $P_{\left(k\right)}^\top P_{\left(k\right)}=I_{n-k}$
such that $P=\Pi_{k=1}^{n-r}P_{\left(k\right)}$. Then
the $\Hinf$ norm of the error reduction system $\Sigma_{e}$ \eqref{eq: reduction error system} is bounded by
\begin{equation}
\| \Sigma_{e}\| _{\Hinf}\leq\sum_{k=1}^{n-1}b\left(\Sigma_{\left(k-1\right)},P_{\left(k\right)}\right),\label{eq:Singleton PROM sequence bound}
\end{equation}
with $b(\Sigma_{(k-1)},P_{(k)})$
given in \eqref{eq:PROM error bound equation}, and 
\begin{equation}
\Sigma_{\left(k\right)}:=\left(P_{\left(k\right)}^\top A_{\left(k-1\right)}P_{\left(k\right)},P_{\left(k\right)}^\top B_{\left(k-1\right)},C_{\left(k-1\right)}P_{\left(k\right)},D\right)
\end{equation}
with $\Sigma_{\left(0\right)}:=\left(A,B,C,D\right)$.
\end{prop}

\begin{proof} We express $\Sigma_{e}\left(s\right)=\Sigma_{r}-\Sigma$
as the telescoping sum $\sum_{k=1}^{n-r}\left(\Sigma_{\left(k\right)}-\Sigma_{\left(k-1\right)}\right)$
with $\Sigma_{\left(0\right)}=\Sigma$ and $\Sigma_{\left(n-r\right)}=\Sigma_{r}$,
such that
\begin{align*}
\| \Sigma_{e}\left(s\right)\| _{\Hinf} & =\| \Sigma_{r}-\Sigma\| _{\Hinf}\\
 & =\| \sum_{k=1}^{n-r}\left(\Sigma_{\left(k\right)}-\Sigma_{\left(k-1\right)}\right)\| _{\Hinf}
\end{align*}
and from the triangle inequality we get
\[
\| \Sigma_{e}\left(s\right)\| _{\Hinf}\leq\sum_{k=1}^{n-r}\| \Sigma_{\left(k\right)}-\Sigma_{\left(k-1\right)}\| _{\Hinf}.
\]

The system $\Sigma_{\left(k\right)}$ has realization 
\begin{align*}
\Sigma_{\left(k\right)} & :=\left(A_{\left(k\right)},B_{\left(k\right)},C_{\left(k\right)},D\right)\\
 & =\left(P_{\left(k\right)}^\top A_{\left(k-1\right)}P_{\left(k\right)},P_{\left(k\right)}^\top B_{\left(k-1\right)},C_{\left(k-1\right)}P_{\left(k\right)},D\right),
\end{align*}
which is an IIPROM of $\Sigma_{\left(k-1\right)}$, therefore, from
Theorem~\ref{PROM error bound},
\[
\| \Sigma_{\left(k\right)}-\Sigma_{\left(k-1\right)}\| _{\Hinf}\leq b\left(\Sigma_{\left(k-1\right)},P_{\left(k\right)}\right),
\]
and we obtain that 
\[
\| \Sigma_{e}\left(s\right)\| _{\Hinf}\leq\sum_{k=1}^{n-r}b\left(\Sigma_{\left(k-1\right)},P_{\left(k\right)}\right).
\]
$\qed$
 \end{proof}

In the following section, we will utilize graph contractions for obtaining suboptimal solutions of Problem \ref{optimal IIPROM bound} for multi-agent systems, and therefore, also Problem \ref{optimal IIPROM}.

\section{Model Reduction of Multi-Agent Systems by Graph Contractions}\label{sec:Model-Reduction-of}

Multi-agent systems may be of extremely large scale, and designing and implementing full order controllers for such systems is not feasible without applying model reduction on the design model or the full-order controller. The general statement of Problem \ref{optimal IIPROM bound} does not suggest any constructive way to find the optimal PROM bound for MAS. However, it is expected that an optimal solution will have some functional dependency on the MAS structure. Vertex partitions have been widely used in graph theory, e.g., for graph clustering \cite{schaeffer2007graph}
and in the study of network communities \cite{newman2004finding}. Vertex partitions have been also used for constructing projection-based
model reductions of multi-agent systems as the consensus protocol
\cite{monshizadeh2014projection} and bidirectional networks \cite{ishizaki2014model}. 

It has been observed in these previous studies that partition-based PROMs maintain an MAS structue, i.e., the PROM $\Sigma_{r}(\Sigma,P)$ is an MAS \eqref{eq:controlled multi-agent system} defined over a reduced order graph $\G_{r}$, $\Sigma\left(\G_{r},\mathcal{U},\Y\right)$.

In this work we introduce the notion of edge-induced PROMs.  These are PROMs which are constructed over edge-induced partitions  of the graph. This graph-based model reduction method allows us to derive sub-optimal but efficient IIPROMs of MAS. These algorithms give good in practice  results as demonstrated in Section \ref{sec:Case-Studies};  however, an  analytic result quantifying  their sub-optimality is yet to be derived.
We first define several combinatorial graph operations that will be
used in this section.

An\textit{ r-partition}, $\pi$, of a vertex set $\V$ is the partition
$\left\{ C_{i}\right\} _{i=1}^{r}$ of $\V$ to $r$ cells such that
$\cup_{i=1}^{r}C_{i}=\V$ and $\left|C_{i}\cap C_{j}\right|=0$ for
$i\neq j$. An $r$-partition $\pi=\left\{ C_{i}\right\} _{i=1}^{r}$ is $\mathcal{\ensuremath{I}}$-invariant
for $\mathcal{\ensuremath{I}}\subseteq\V$ if $\forall i\in\left[1,r\right]$, one has $\left|C_{i}\cap\mathcal{\ensuremath{I}}\right|\leq1$, i.e., no partition
cell contains more than one node in $\mathcal{\ensuremath{I}}$.
The $r$-partition is \emph{strongly} $\mathcal{\ensuremath{I}}$-invariant if all partition
cells containing nodes in $\mathcal{\ensuremath{I}}$ are singletons, i.e., $\left|C_{i}\right|=1$ whenever $C_i \subset \mathcal{I}$.  
We denote the set of all strongly $\mathcal{\ensuremath{I}}$-invariant  $r$-partitions as $S_{r}\left(\V\right)$. Figure \ref{fig:An interface invariant partition} illustrates these definitions.
\begin{figure}[b]
\centering
\subfloat[][]{\includegraphics[width=1.0in]{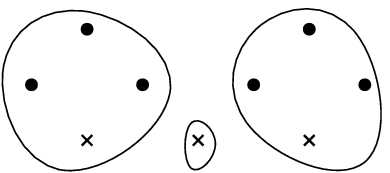}}
\hfill
\subfloat[][]{\includegraphics[width=1.0in]{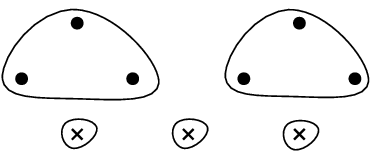}}
\hfill
\subfloat[][]{\includegraphics[width=1.0in]{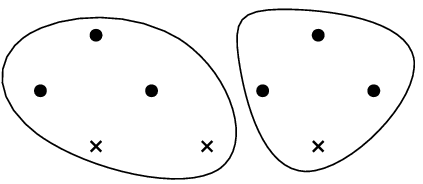}}
\caption{A set of $9$ nodes, with subset $\mathcal{\ensuremath{I}}$  of three  nodes (marked with x), (a) is an $\mathcal{\ensuremath{I}}$-invariant  $3$-partition, (b) is a strongly $\mathcal{\ensuremath{I}}$-invariant  $5$-partition, and (c) is a $2$-partition which is not $\mathcal{\ensuremath{I}}$-invariant.
\label{fig:An interface invariant partition}}
\end{figure}

Given a graph $\graph$, and an $n-r$ edge subset $\E_{S}\subseteq\E$,
an \emph{edge-induced partition} $\pi(\E_{S})$ is an $r$-partition of $\V$ constructed as follows \cite{leiter2021edge}: (i) A graph $\G(\V,\E_{S})$ is created from the vertices of $\G$ and the edge-subset $\E_{S}$, (ii) the connected components of $\G(\V,\E_{S})$ are found, (iii) the vertices of each component is registered as a partition cell, and the set of all components cells constitutes the partition $\pi(\E_{S})$ of $\V$ (Figure \ref{fig:edge partition}).

Given an $r$-partition $\pi=\left\{ C_{i}\right\} _{i=1}^{r}$, we
define $M\left(\pi\right)\in\mathbb{\mathbb{R}}^{n\times r},$ the
\emph{partition characteristic matrix} (PCM) with entries $\left[M\left(\pi\right)\right]_{ij}=1$
if $i\in C_{j}$, and $0$ otherwise. The corresponding \emph{partition
projection matrix} (PPM) is $P\left(\pi\right)\triangleq M\left(\pi\right)D^{-\frac{1}{2}}\left(\pi\right)$
where $D\left(\pi\right)\triangleq M^\top\left(\pi\right)M\left(\pi\right)$

For multi-agent systems we restrict the projection to be PPMs, and
the partition IIPROM bound problem is introduced:

\begin{figure}
\centering
\subfloat[][]{\includegraphics[width=.8in]{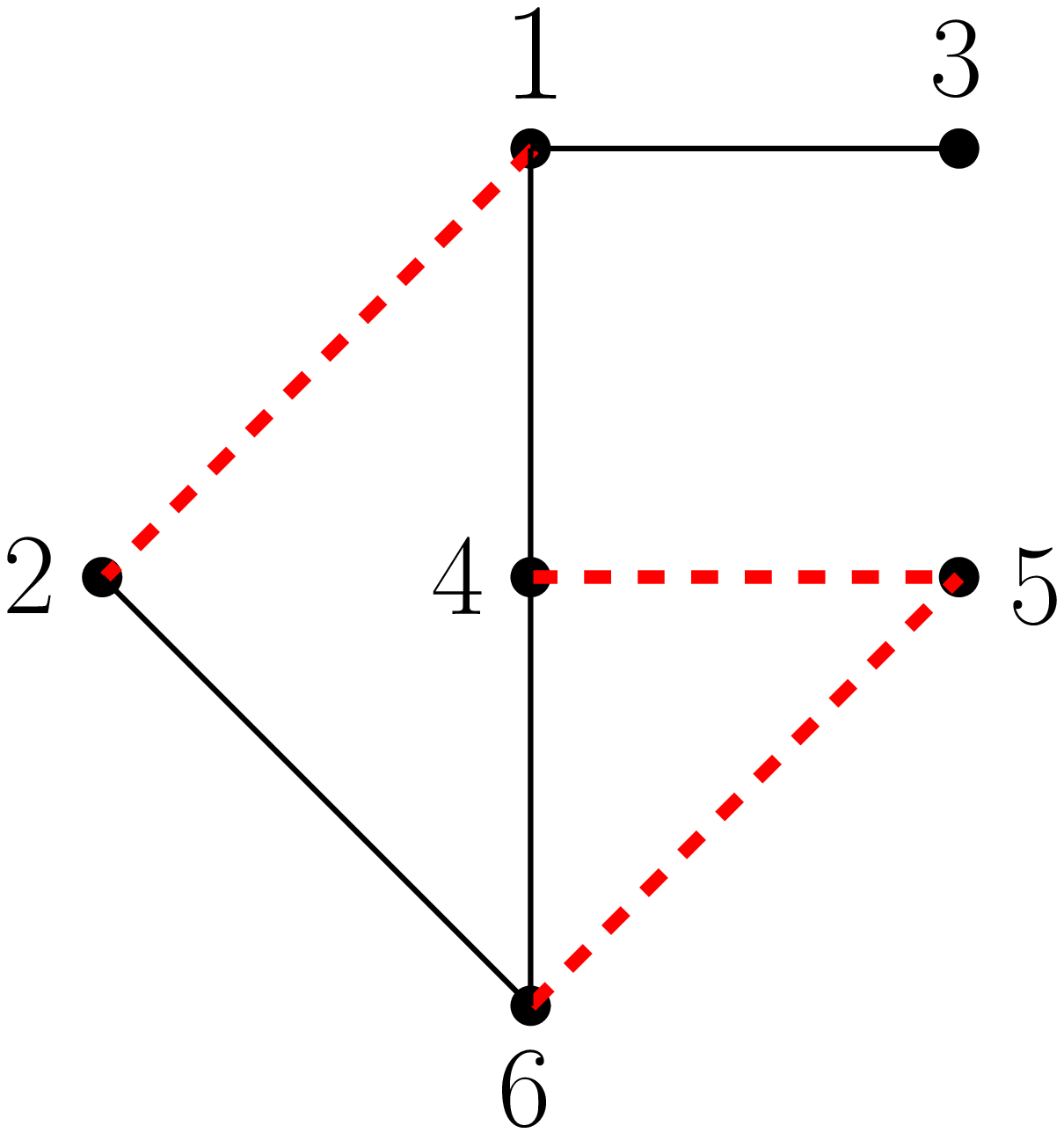}}\hspace{1cm}
\subfloat[][]{\includegraphics[width=.8in]{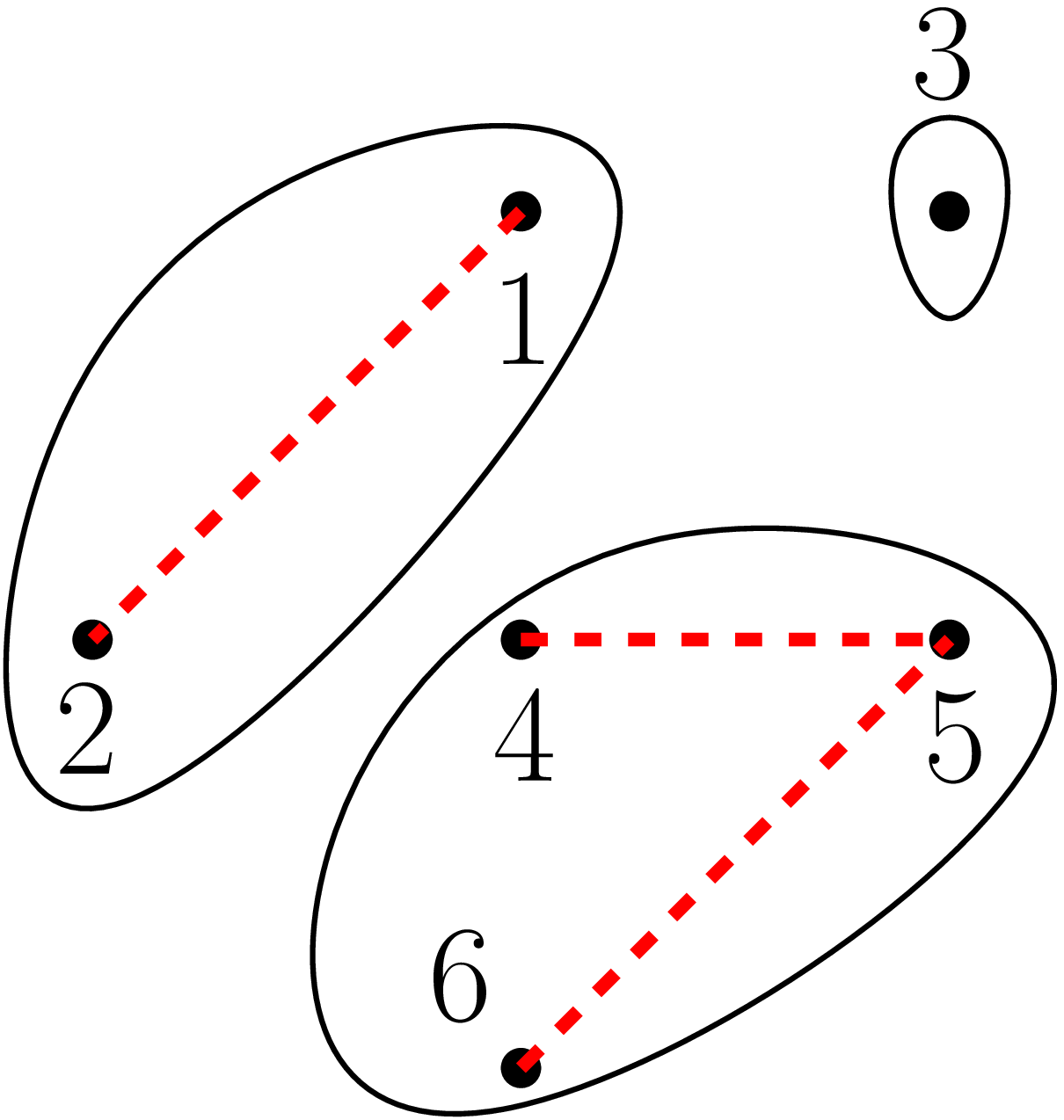}}
\caption{An edge-induced partition of a graph of order $r$, (a) a graph $\graph$ with a selected edge subset $\E_S=\{\{1,2\}, \{4,5\}, \{5,6\}\}$ (dashed red), (b) the graph $\G(\V,\E_S)$ and its connected components inducing the partition $\pi(\E_S)=\{\{1,2\},\{3\},\{4,5,6\}\}$ on $\V(\G)$.
\label{fig:edge partition}}
\end{figure}

\begin{prob}[optimal partition IIPROM bound]\label{optimal partition IIPROM bound}Consider
a stable MAS $\Sigma\left(\G,\mathcal{U},\Y\right)$ with $n$ agents,
each with local state of dimension $d_{x}$, and an interface set
$\mathcal{I}=\U\cup\Y$. Find an $\mathcal{I}$-invariant $r$-partition
$\pi$ such that the PROM $\Sigma_{r}\left(\Sigma,P\left(\pi\right)\right)$
minimizes the reduction error bound $b\left(\Sigma\left(\G,\mathcal{U},\Y\right),I_{d_{x}}\otimes P\left(\pi\right)\right)$ given in 
\eqref{eq:PROM symmetric error bound equation}.\end{prob}

Finding a solution to Problem \ref{optimal partition IIPROM bound}
may be numerically intractable for a moderate number of nodes, as
the number of $r$-partitions is the Stirling number of the second
kind, $$S\left(n,r\right)=\sum\limits _{k=1}^{r}\left(-1\right)^{r-k}\frac{k^{n}}{k!\left(r-k\right)!},$$ which for $r\ll n$ is asymptotically $S\left(n,r\right)\sim\frac{r^{n}}{r!}$
\cite[p.18]{wilf1994generatingfunctionology}. 

Given a subset of edges $\E_{c}\subseteq\E$, an \emph{edge-induced partition} $\pi\left( \E_{c} \right)$ can be constructed as described above. Here we utilize the edge-induced partition to derive a greedy-edge IIPROM bound (GEIB) algorithm (Algorithm \ref{greedy-edge-IIPROM}). The input to the algorithm is an MAS $\Sigma\left(\G,\U,\Y\right)$ \eqref{eq:controlled multi-agent system}, the required reduction order $r$ and a subset of candidate edges $\E_{c}\subseteq\E$ (assuming there are more than $n-r$ edges in $\E_{c}$). The first step of the algorithm is to check if each of the edges is strongly $\mathcal{I}$-invariant by examining if both end nodes of an edge are not in  $\mathcal{I}$. If an edge is found to be strongly $\mathcal{I}$-invariant, a PPM in constructed with its induced edge partition and the PROM error bound \eqref{PROM error bound} is calculated. The PPM of the induced-edge partition of the $n-r$ edges with lowest bound is then used to obtain a PROM and serves as a suboptimal solution to Problem \ref{optimal partition IIPROM bound}. 

\begin{center}
\begin{algorithm}[h]
\caption{Greedy-edge IIPROM bound Algorithm} \label{greedy-edge-IIPROM}

\textbf{Input: }An MAS $\Sigma\left(\G,\U,\Y\right)$ of order $n$
with realization $\Sigma:=\left(A\left(\G\right),B\left(\G,\mathcal{U}\right),C\left(\G,\Y\right),D\left(\U,\Y\right)\right)$,
interface set $\mathcal{I}=\U\cup\Y$, edge subset $\E_{c}\subseteq\E$,
reduction order $r$.
\begin{enumerate}
\item \textbf{For }each $\left\{ u,v\right\} \in\E_{c}$
\begin{itemize}
\item \textbf{if} $u\in\mathcal{I}$ or $v\in\mathcal{I}$, skip to next edge,

\item \textbf{else} construct the edge-induced partition $\pi\left( \left\{ u,v\right\} \right)$ and its PPM $P\left(\pi\right)$, and calculate  $b\left(\Sigma,P\left(\pi\right)\right).$
\end{itemize}
\item \textbf{Find }$n-r$ edges $\E^{*}\subseteq\E_{c}$ with the lowest bound values.
\item \textbf{Construct }the edge-induced $r$-partition $\pi\left(\E^{*}\right)$.
\end{enumerate}
\textbf{Output:} $\Sigma_{r}\left(\Sigma,P\left(\pi\left(\E^{*}\right)\right)\right)$
\eqref{eq:PROM}

\textbf{Notation: $\Sigma_{r}=GEIB\left(\Sigma\left(\G,\U,\Y\right),\E_{c},r\right)$}%
\end{algorithm}
\par\end{center}

The greedy-edge IIPROM bound algorithm does not specify the method to choose the edge subset $\E_{c}$. Trees are the building blocks of any connected graph. A basic graph-theoretic principle is that a spanning tree of a connected graph $\graph$ of order $n$ is a subgraph  with a minimal set of $n-1$ edges connecting all vertices. Furthermore, trees and cycle-completing edges are strongly related to the performance of networked systems \cite{Zelazo2011}. With this intuition, we derive the tree-based IIPROM bound (TBIB) algorithm (Algorithm \ref{Tree-based IIPROM bound Algorithm}), where given an MAS $\Sigma\left(\G,\U,\Y\right)$ \eqref{eq:controlled multi-agent system},  a spanning-tree is found. The edges of the tree are then used as the subset $\E_{c}$ when applying the GEIB algorithm.

\begin{center}
\begin{algorithm}[h]
\caption{Tree-based IIPROM bound Algorithm}
\label{Tree-based IIPROM bound Algorithm}
\textbf{Input: }An MAS $\Sigma\left(\G,\U,\Y\right)$ of order $n$
with realization $\Sigma:=\left(A\left(\G\right),B\left(\G,\mathcal{U}\right),C\left(\G,\Y\right),D\left(\U,\Y\right)\right)$,
interface set $\mathcal{I}=\U\cup\Y$, reduction order $r$.
\begin{enumerate}
\item \textbf{Find} a spanning tree $\T$ of $\G$.
\item \textbf{Perform} greedy edge algorithm $$\Sigma_{r}=GEIB\left(\Sigma\left(\G,\U,\Y\right),\E\left(\T\right),r\right).$$
\end{enumerate}
\textbf{Output:} $\Sigma_{r}$

\textbf{Notation: $\Sigma_{r}=TBIB\left(\Sigma\left(\G,\U,\Y\right),r\right)$}%

\end{algorithm}
\par\end{center}

We have derived the TBIB algorithm utilizing the PROM error bound as a general framework for model reduction of MAS. In the following case study section, we examine the performance of the algorithm, and it will be shown to provide excellent results for the reduction of large-scale MAS.

\section{Case Studies \label{sec:Case-Studies}}

In this section we present some numerical examples illustrating the
results of this work. 

The Laplacian matrix plays a key role in networked an multi-agent systems. For a multi-graph, i.e., a graph that may include duplicate edges and self loops, the Laplacian matrix $L\left(\G\right)\in\R^{\left|\V\right|\times\left|\V\right|}$
is defined as \cite{Godsil2001}
\begin{equation}
\left[L(\mathcal{G})\right]_{uv}=\begin{cases}
\sum_{v\in\left\{ \N_{u}\cup u\right\} }w\left(\left\{ u,v\right\} \right) & u=v\\
-w\left(\left\{ u,v\right\} \right) & u\sim v\\
0 & o.w.
\end{cases},
\end{equation}
where $w(\{u,v\})$ denotes the weight of the edge $\{u,v\}$.

As a test case model of an interface
MAS of the form (\ref{eq:controlled multi-agent system}), we consider the \emph{Laplacian controlled consensus
model (LCC)} \eqref{eq:controlled consensus protocol}, a generalisation of the Laplacian Consensus protocol $\dot{x}=-L\left(\G\right)x$,  a benchmark model of an uncontrolled multi-agent system 
\cite{mesbahi2010graph},
\begin{equation}
\begin{cases}
\dot{x} & =-L\left(\G\right)x+B\left(\mathcal{U}\right)P\\
y & =C\left(\mathcal{Y}\right)x+D\left(\U,\Y\right)u
\end{cases}.\label{eq:controlled consensus protocol}
\end{equation}
In the following subsections we will demonstrate the  effectiveness of the graph-based model reduction Algorithms derived in Section \ref{sec:Model-Reduction-of} for the reduction of small and large-scale LCC models.

\subsection{IIPROM of a Small-Scale Lapalcian Consensus Model}
\begin{figure*}
\centering
\subfloat[][]{\includegraphics[width=2.25in]{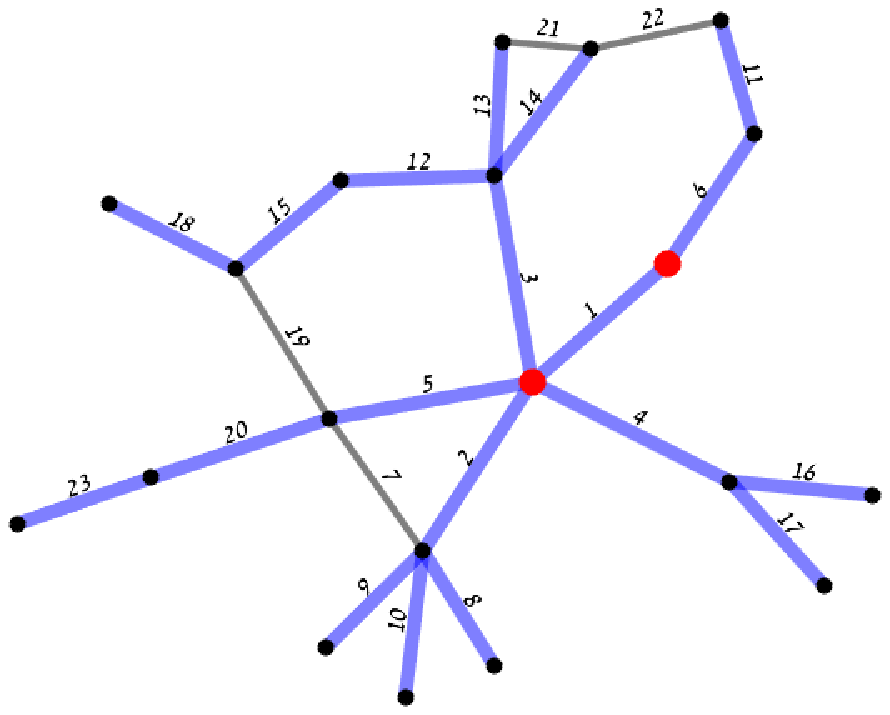}\label{fig:The-Laplacian-consensus case study 1}}
\subfloat[][]{\includegraphics[width=2.25in]{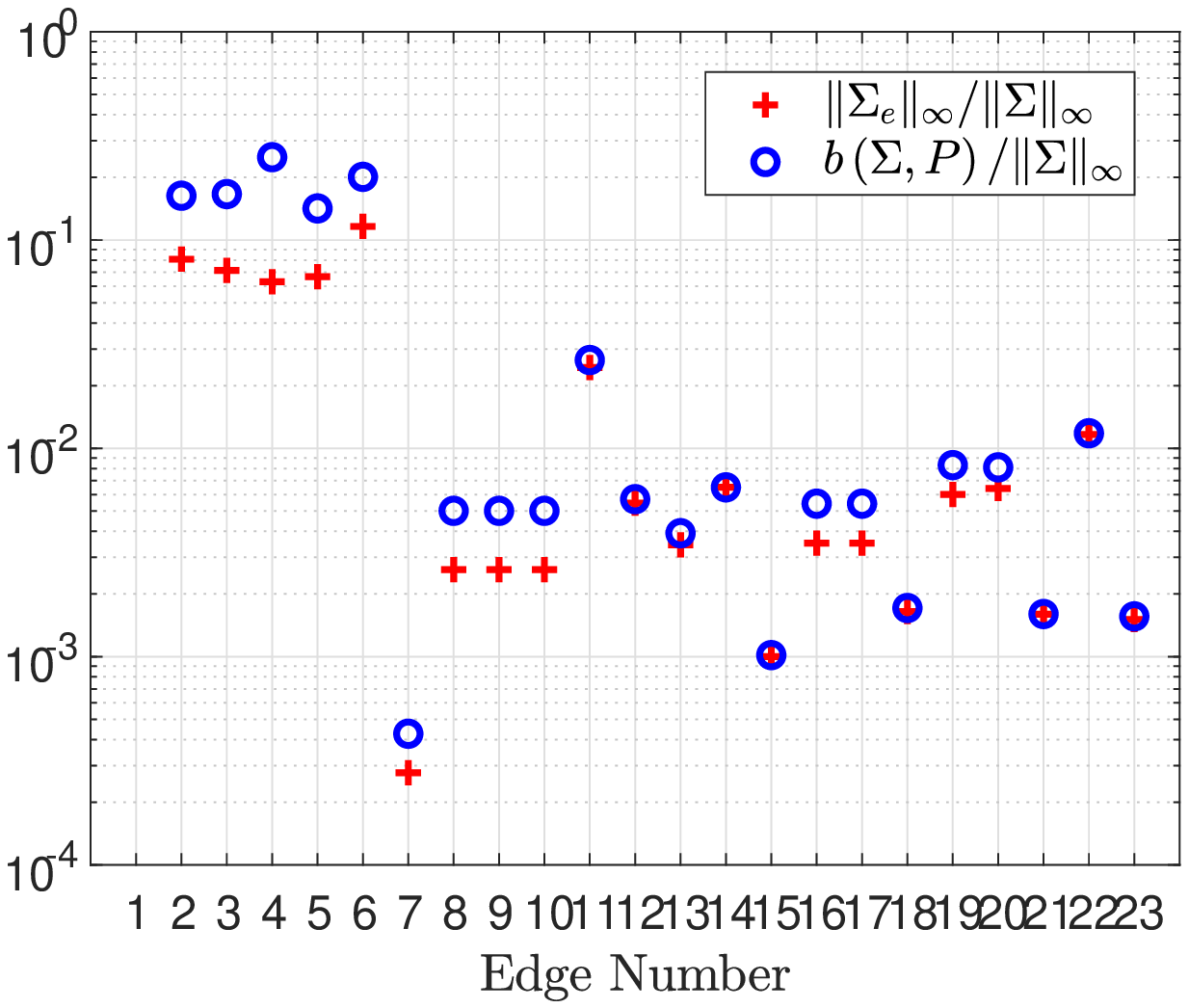}\label{fig:The-singleton iiprom error case study 1}}
\subfloat[][]{\includegraphics[width=2.25in]{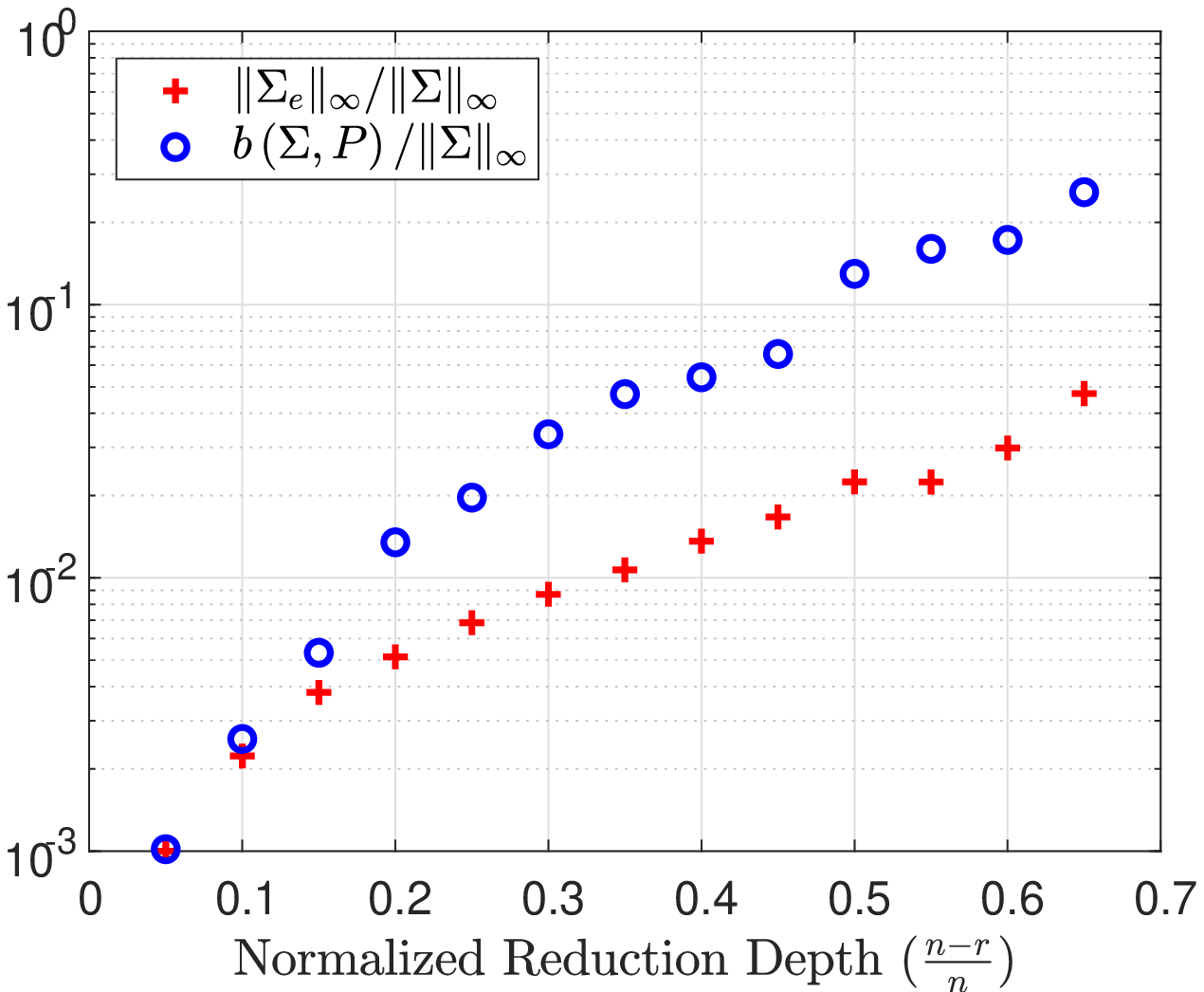}\label{fig:The-Laplacian-consensus tree case study error}}
\caption{A small-scale edge-based IIPROM case study (a) The graph associated to the Laplacian consensus in this example, interface nodes marked red, labels on edges indicate the edge numbers, a spanning tree is highlighted in blue and (b) The reduction error and reduction error bound of each edge-based IIPROM of the Laplacian consensus
associated to the graph in Figure \ref{fig:The-Laplacian-consensus case study 1}. Notice that no values are assigned to edge $1$ since it does not induce an interface invariant PROM. (c) The reduction error and reduction error bound resulting from applying the tree-based  IIPROM  bound  algorithm to the the Laplacian consensus
associated to the graph.}
\end{figure*}

As a first case study we consider the reduction of an LCC model (\ref{eq:controlled consensus protocol})
associated to a small-scale graph $\G$ with $20$ nodes and $23$ edges with
$2$ interface nodes (Figure \ref{fig:The-Laplacian-consensus case study 1}).
We perform an edge-based IIPROM for each of the $\mathcal{\ensuremath{I}}$-invariant singleton partitions induced by the edges of $\G$. In this case, only the partitions induced by edge $1$ are not $\mathcal{\ensuremath{I}}$-invariant since both its end nodes are interface nodes. 
For each of these IIPROMs, the $\Hinf$-norm of the error system \eqref{PROM error system} and the $\Hinf$ error bound \eqref{PROM error bound} are calculated and normalized by $\| \Sigma\| _{\Hinf}$ (Figure \ref{fig:The-singleton iiprom error case study 1}). 
We observe that the choice of edge has a great effect on the reduction error magnitude, some edges produce reduction errors smaller than $1\%$ of $\| \Sigma\| _{\Hinf}$, while other edges induce IIPROM errors that are $10\%$ of $\| \Sigma\| _{\Hinf}$. It is also observed that the error bound is relatively tight for small reduction errors, while for larger errors the bound may differ more than an order of magnitude from the error.  

We then demonstrate the tree-based IIPROM bound algorithm (Algorithm \ref{Tree-based IIPROM bound Algorithm}) on the LCC. The spanning tree edges are found and Algorithm \ref{Tree-based IIPROM bound Algorithm} is then performed for $r\in\left[10,19\right]$. Figure \ref{fig:The-Laplacian-consensus tree case study error}
plots the reduction error $\| \Sigma_{e}\| _{\Hinf}$,
and the error bound $b_{r}\left(\G,U_{r}\right)$ \eqref{eq:PROM error bound equation},
as a function of the normalized reduction depth $\frac{n-r}{n}$. We observe that the log reduction
error has a quasi-linear trend as a function of $\frac{n-r}{n}$, and that the bound is tight for low reduction depths and differs as the reduction increases; however, the bound follows the same trend as the error, therefore, minimizing the bound is consistent with minimizing the error in this case.


\subsection{IIPROM of A Large-Scale Small-World Laplacian Consensus Model}
As a large-scale case study, a small-world graph is created with the Watts-Strogatz random rewiring procedure with $k=5$ and $\beta=0.15$ \cite{watts1998collective} starting from a $5$-regular graph of order $100$ with $5$ interface nodes $\U=\Y=\{1,\ldots,5\}$ (Figure \ref{fig:small world case study}). An LCC is constructed over this graph and the tree-based IIPROM bound algorithm (Algorithm \ref{Tree-based IIPROM bound Algorithm}) is then performed. Figure \ref{fig: small world PROM error} plots the reduction error $\| \Sigma_{e}\| _{\Hinf}$,
and the error bound $b_{r}\left(\G,U_{r}\right)$ \eqref{eq:PROM error bound equation}  (normalized by $\| \Sigma\| _{\Hinf}$),
as a function of the normalized reduction depth $\frac{n-r}{n}$. 
As a comparison to the reduction error and bound  results, we calculate the empirical mean-IIPROM $\mu_{P}\left(\Sigma\right)$, the mean reduction error $\| \Sigma_{e}\| _{\Hinf}$ of $N=50$ randomly selected IIPROMs, and the empirical mean edge-based IIPROM $\mu_{\varepsilon}\left(\Sigma\right)$, the mean reduction error $\| \Sigma_{e}\| _{\Hinf}$ of $N=50$ randomly selected edge-based PPM IIPROMs. 

We observe that bound is tight for the entire reduction depth range. Furthermore, the reduction with the proposed tree-based method is several orders of magnitude lower than the empirical mean-IIPROM, and for lower reduction depth is significantly better than the empirical mean edge-based IIPROM. As expected, for high reduction depth, the tree-based method converges to the empirical mean edge-based IIPROM.    

\begin{figure}
\centering
\includegraphics[width=.55\columnwidth]{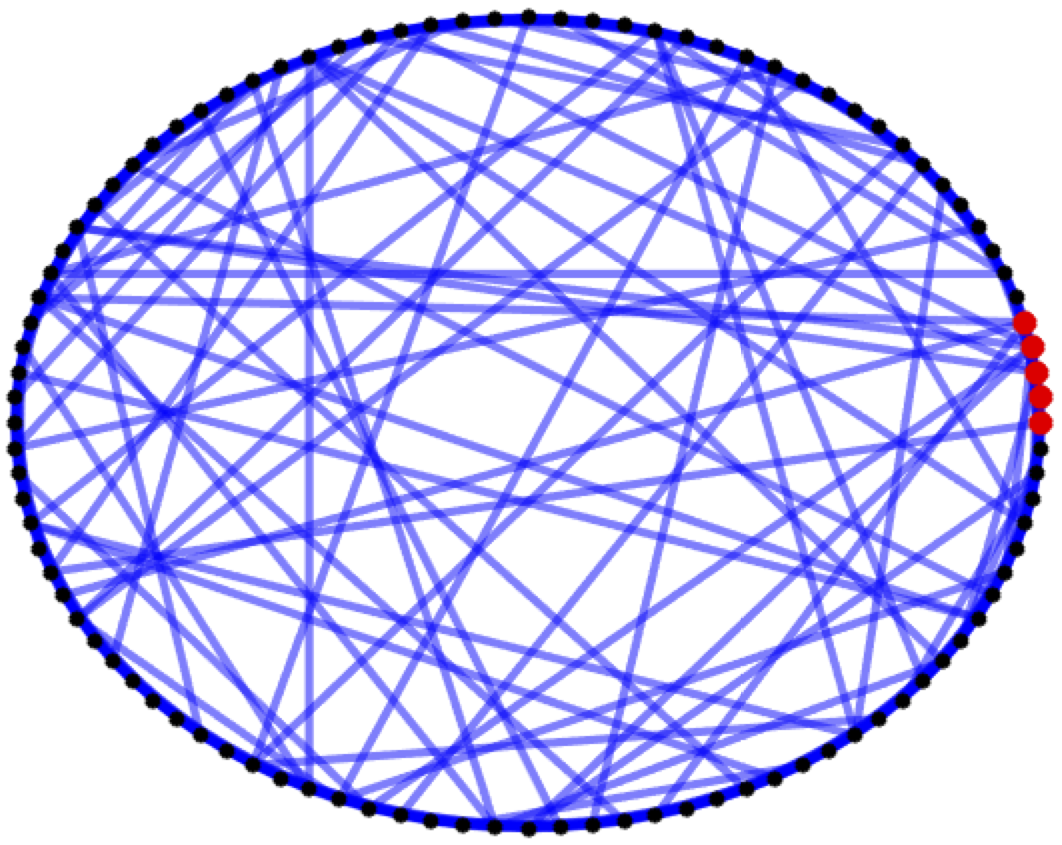}\caption{The Laplacian consensus in this example is associated to a small-world graph with $5$ interface nodes (marked red). }\label{fig:small world case study}
\end{figure}

\begin{figure}
\centering
\includegraphics[width=.55\columnwidth]{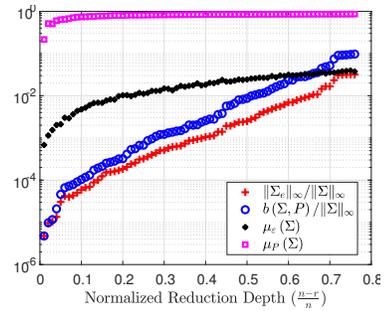}\caption{The reduction error and reduction error bound of the Laplacian consensus
associated to the graph in Figure \ref{fig:small world case study}. }\label{fig: small world PROM error}
\end{figure}

\section{Conclusions\label{sec:Conclusions}}

This work derived a unique product form of the  error system of orthogonal projection-based reduced models. This product form is then used to derive an $\Hinf$ error bound for the PROM error system. A suboptimal bound optimization solution for multi-agent systems is obtained with a graph-based spanning tree algorithm.   
Applying this algorithm on a Laplacian consensus model constructed with a large-scale small-world  network, demonstrates the utility of the method to large scale multi-agent systems.
In the examined case studies, the bound obtained with the tree-base algorithm is tight, therefore, the optimal PROM bound for those cases is close to the optimal reduction error PROM. It is an open research question to explain why the bound tight.
The same technique presented can be applied to various multi-agent systems other than the consensus models. 
The derived reduction error product form and  bound can be the basis for additional optimization methods, such as convex relaxations.
%

\bibliographystyle{plain}        
\bibliography{thesis_refs}           

\appendix
\section{Proof of Theorem \ref{PROM error system}}
In this Appendix, we present the proof of Theorem \ref{PROM error system}.
The proof is based on the matrix inverse lifting lemma along with a projection inversion corollary that we present here.

\begin{defn}[Projected Schur complement]\label{-Let-Schur complement}
Let $M\in\mathbb{C}^{n\times n}$ and let $P$ and $Q$ be projections
such that $P^\top P=I_{r}$, $Q^\top Q=I_{n-r}$ and
$Q^\top P=0$. Then for $P^\top MP$ invertible, we define the \emph{projected
Schur complement} of $S$ by $P$ as
\begin{equation}
S\left(M,P\right)\triangleq M_{QQ}-M_{QP}M_{PP}^{-1}M_{PQ},
\end{equation}
and the \textit{projected Schur complement} of $S$ by $Q$ as
\begin{equation}
S\left(M,Q\right)\triangleq M_{PP}-M_{PQ}M_{QQ}^{-1}M_{QP},
\end{equation}
where $M_{PP}\triangleq P^\top MP$, $M_{PQ}\triangleq P^\top MQ$, $M_{QP}\triangleq Q^\top MP$
and $M_{QQ}\triangleq Q^\top MQ$.

\end{defn}

With the definition of the projected Schur complement we can derive the following corollary of the matrix inversion Lemma \cite{bernstein2009matrix}.

\begin{lem}
Let $M\in\mathbb{C}^{n\times n}$ and $M^{-1}$ be a matrix and its inverse with corresponding block structures 
\begin{equation*}
M=\left[\begin{array}{cc}
M_{11} & M_{12}\\
M_{21} & M_{22}
\end{array}\right]
\,\,\,\,\,\,\,\,\,\,\,\,\,\,
M^{-1}=\left[\begin{array}{cc}
W & X\\
Y & Z
\end{array}\right],
\end{equation*}
then 
\begin{align}
W & = (M_{11}-M_{12}M_{22}^{-1}M_{21})^{-1},
\end{align}
and 
\begin{equation}
X = -M_{11}^{-1}M_{12}Z.
\end{equation}
\end{lem}

\begin{cor}\label{projection inversion corollary} Consider a
matrix $M\in\mathbb{C}^{n\times n}$, and let $P$ and $Q$ be matrices
such that $P^\top P=I_{r}$, $Q^\top Q=I_{n-r}$ and
$Q^\top P=0$. Then
\begin{align}
P^\top M^{-1}P & =S^{-1}\left(M,Q\right),
\end{align}
and 
\begin{equation}
P^\top M^{-1}Q=-M_{PP}^{-1}M_{PQ}\left(Q^\top M^{-1}Q\right).
\end{equation}
\end{cor}
\begin{proof} Apply
the matrix inversion Lemma  to the matrix
$\tilde{M}=\left[PQ\right]^\top M\left[PQ\right]$ and obtain $P^\top  M^{-1}P=S^{-1}\left(M,Q\right)$ and $P^\top M^{-1}Q=-M_{PP}^{-1}M_{PQ}\left(Q^\top M^{-1}Q\right)$. $\qed$
\end{proof}

We denote a \emph{matrix lifting} $f_P:\R^{r\times r}\rightarrow\R^{n\times n}$ as the function $f_P\left( M \right) \triangleq PMP^\top$. The following lemma extends the matrix inversion lemma to the matrix
lifting 
of $\left(P^\top MP\right)^{-1}$.

\begin{lem}[Matrix Inverse Lifting]\label{The Matrix Lifting Lemma}
Consider a  matrix $M\in\mathbb{C}^{n\times n}$ and let $P$
and $Q$ be matrices such that $P^\top P=I_{r}$, $Q^\top Q=I_{n-r}$
and $Q^\top P=0$. Then
\begin{equation}
P\left(P^\top MP\right)^{-1}P^\top=\Upsilon\left(M,Q\right),\label{eq:matrix projection lifting}
\end{equation}
where we define $\Upsilon\left(M,Q\right)$ as,
\[
\Upsilon\left(M,Q\right)\triangleq M^{-1}-M^{-1}Q\left(Q^\top M^{-1}Q\right)^{-1}Q^\top M^{-1}.
\]
\end{lem}
\begin{proof}We have $I_{n}=PP^\top+QQ^\top$ such that
\begin{align*}
\Upsilon\left(M,Q\right) & =\left(PP^\top+QQ^\top\right)\Upsilon\left(M,Q\right)\left(PP^\top+QQ^\top\right)\\
 &\hspace{-.5cm} =PP^\top\Upsilon\left(M,Q\right)PP^\top +PP^\top\Upsilon\left(M,Q\right)QQ^\top\\
 &\hspace{-.5cm} +QQ^\top\Upsilon\left(M,Q\right)PP^\top +QQ^\top\Upsilon\left(M,Q\right)QQ^\top.
\end{align*}
Evaluating the four expressions in the sum we get
\begin{align*}
&P^\top\Upsilon\left(M,Q\right)Q  =\\ 
& P^\top M^{-1}Q- P^\top M^{-1}Q\left(Q^\top M^{-1}Q\right)^{-1}Q^\top   M^{-1}Q =0,\\
&Q^\top\Upsilon\left(M,Q\right)P  =\\
&\ Q^\top M^{-1}P- Q^\top M^{-1}Q\left(Q^\top M^{-1}Q\right)^{-1}Q^\top   M^{-1}P =0,\\
&Q^\top\Upsilon\left(M,Q\right)Q  =\\
& Q^\top M^{-1}Q- Q^\top M^{-1}Q\left(Q^\top M^{-1}Q\right)^{-1}Q^\top   M^{-1}Q =0,\\
&P^\top\Upsilon\left(M,Q\right)P  =\\
& P^\top M^{-1}P- P^\top M^{-1}Q\left(Q^\top M^{-1}Q\right)^{-1}Q^\top   M^{-1}P.
\end{align*}
We observe that the last term $P^\top\Upsilon\left(M,Q\right)P$ 
is the projected Schur complement $S\left(M^{-1},Q\right)$ (Def. \ref{-Let-Schur complement}).
From the projection inversion corollary (Corollary \ref{projection inversion corollary})
we then obtain
\begin{align*}
P^\top\Upsilon\left(M,Q\right)P & =\left(P^\top MP\right)^{-1},
\end{align*}
and therefore
\[
\Upsilon\left(M,Q\right)=P\left(P^\top MP\right)^{-1}P^\top.
\]
$\qed$
\end{proof}

We are now prepared to proceed with the proof of Theorem \ref{PROM error system}.
\begin{proof} We begin by proving the first part of the theorem, i.e.,
\begin{equation}
\hat{\Sigma}_{e}(\Sigma,P)=C\Phi^{-1}Q(Q^\top\Phi^{-1}Q)^{-1}Q^\top\Phi^{-1}B.\label{eq:PROM error TFM appendix}
\end{equation}
The error system TFM is 
\begin{align*}
\hat\Sigma_{e} & =\hat\Sigma_{r}-\hat\Sigma\\
 & =CP\left(sI_{r}-P^\top AP\right)^{-1}P^\top B-C\Phi^{-1}B\\
 & =CP\left(P^\top\Phi P\right)^{-1}P^\top B-C\Phi^{-1}B\\
 & =C\left(P\left(P^\top\Phi P\right)^{-1}P^\top-\Phi^{-1}\right)B.
\end{align*}
We now employ the matrix inverse lifting lemma (Lemma \ref{The Matrix Lifting Lemma}),
%
\begin{align*}
P(P^\top\Phi P)^{-1}P^\top-\Phi^{-1} &
=\Phi^{-1}Q(Q^\top\Phi^{-1}Q)^{-1}Q^\top\Phi^{-1},
\end{align*} 
thus leading to the expression 
\eqref{eq:PROM error TFM appendix}.
Next, we prove the second part of the theorem, i.e.,
\begin{equation}
\Sigma_{e}(\Sigma,P)=\Theta(\Sigma,P)\Delta(\Sigma,P)\Gamma(\Sigma,P),
\end{equation} 
with the three realizations  
\begin{align*}
\Theta(\Sigma,P)&:=(A_{PP},A_{PQ},CP,CQ)\\
\Gamma(\Sigma,P)&:=(A_{PP},P^\top B,A_{QP},Q^\top B)\\
\Delta(\Sigma,P)&:=(A,Q,Q^\top,0_{p\times m}).
\end{align*} 
With $I=\left(Q^\top\Phi^{-1}Q\right)\left(Q^\top\Phi^{-1}Q\right)^{-1}$ we get
\begin{align*}
&C\Phi^{-1}Q \Xi(\Phi,Q) Q^\top\Phi^{-1}B  =\\
&\hat{\Theta}\left(\Sigma,P\right)\hat{\Delta}\left(\Sigma,P\right)\hat{\Gamma}\left(\Sigma,P\right)
\end{align*}
where $\Xi(\Phi,Q)\triangleq\left(Q^\top\Phi^{-1}Q\right)^{-1}$ and
\begin{align*}
\hat{\Theta}\left(\Sigma,P\right)&\triangleq C\Phi^{-1}Q\Xi(\Phi,Q),\\
\hat{\Delta}\left(\Sigma,P\right)&\triangleq Q^\top\Phi^{-1}Q,\\\hat{\Gamma}\left(\Sigma,P\right)&\triangleq\Xi(\Phi,Q)Q^\top\Phi^{-1}B.
\end{align*}
We have $I=PP^\top+QQ^\top$ such that
\begin{align*}
\Phi^{-1}Q\left(Q^\top\Phi^{-1}Q\right)^{-1} & =\left(PP^\top+QQ^\top\right)\Phi^{-1}Q\Xi(\Phi,Q)\\
 & =Q+P\left(P^\top\Phi^{-1}Q\right)\Xi(\Phi,Q).
\end{align*}
From the projection inversion corollary (Corollary \ref{projection inversion corollary}), we have
%
\begin{align*}
P^\top\Phi^{-1}Q & =-\Phi_{PP}^{-1}\Phi_{PQ}\left(Q^\top\Phi^{-1}Q\right)\\
 & =\left(sI_{r}-P^\top AP\right)^{-1}P^\top AQ\left(Q^\top\Phi^{-1}Q\right)
\end{align*}
and
\begin{equation}
\Phi^{-1}Q\Xi(\Phi,Q)=Q+P\left(sI_{r}-P^\top AP\right)^{-1}P^\top AQ,
\end{equation}
such that 
\[
\hat{\Theta}(\Sigma,P)=CP\left(sI_{r}-P^\top AP\right)^{-1}P^\top AQ+CQ,
\] which is the TFM with realization $\Theta(\Sigma,P):=\left(A_{PP},A_{PQ},CP,CQ\right)$.
Similarly we get
\begin{align*}
\hat{\Gamma}(\Sigma,P) & =\left[Q^\top\Phi^{-1}Q\right]^{-1}Q^\top\Phi^{-1}B\\
 & =Q^\top AP\left(sI_{r}-P^\top AP\right)^{-1}P^\top B+Q^\top B,
\end{align*}
and the realization $\Gamma(\Sigma,P):=\left(A_{PP},P^\top B,A_{QP},Q^\top B\right)$.
Finaly $\hat{\Delta}(\Sigma,P)=Q\Phi^{-1}Q$
has realzation $\Delta(\Sigma,P):=\left(A,Q,Q^\top,0_{p\times m}\right)$. $\qed$
 \end{proof}

\end{document}